\documentclass[reqno, eucal]{amsart}

\usepackage[mathscr]{eucal} 

\usepackage[dvips]{graphicx}
\usepackage[dvips]{color}
\usepackage{amsmath,amsfonts,amssymb,amsthm,amscd}
\usepackage{epic}
\usepackage{eepic}	
\usepackage{longtable}
\usepackage{array}

\newtheorem{theorem}{Theorem}[section]

\newtheorem{proposition}[theorem]{Proposition}

\theoremstyle{definition}

\newtheorem{example}[theorem]{Example}
\newtheorem{remark}[theorem]{Remark}

\newcommand{\Z}{{\mathbb Z}}

\newcommand{\ok}{{\rm{\bf k}}}
\newcommand{\OK}{{\rm{\bf K}}}
\newcommand{\am}{{\rm{\bf a}}^{\!-} }
\newcommand{\ap}{{\rm{\bf a}}^{\!+} }
\newcommand{\Am}{{\rm{\bf A}}^{\!\!-} }
\newcommand{\Ap}{{\rm{\bf A}}^{\!\!+} }

\begin{document}

\title[3D Reflection equation]
{A solution of the 3D reflection equation from\\
quantized algebra of functions of type $B$}

\author{Atsuo Kuniba}
\address{Institute of Physics\\
Graduate School of Arts and Sciences University of Tokyo\\
Komaba, Tokyo 153-8902, Japan\\
E-mail: atsuo@gokutan.c.u-tokyo.ac.jp}

\author{Masato Okado}
\address{Department of Mathematical Science\\
Graduate School of Engineering Science, Osaka University\\
Toyonaka, Osaka 560-8531, Japan\\
E-mail: okado@sigmath.es.osaka-u.ac.jp}

\maketitle

\begin{abstract}
Let $A_q(\mathfrak{g})$ be the quantized algebra of functions 
associated with simple Lie algebra $\mathfrak{g}$
defined by generators obeying the so called $RTT$ relations.
We describe the embedding $A_q(B_2) \hookrightarrow A_q(C_2)$ explicitly.
As an application, a new solution of the Isaev-Kulish 3D reflection equation
is constructed by combining the embedding with the previous solution for $A_q(C_2)$ 
by the authors.
\end{abstract}



\section{Introduction}\label{sec:1}
The reflection equation \cite{Ch84,Sk88} 
is a companion system of the Yang-Baxter equation \cite{Bax}
in 2 dimensional (2D) integrable systems with boundaries.
A 3D analogue of such a system is known by Isaev and Kulish \cite{IK} 
as the 3D reflection equation 
which accompanies the Zamolodchikov tetrahedron equation \cite{Zam80}.
In \cite{KO2}, the first solution of the 3D reflection equation was 
constructed by invoking the representation theory \cite{So1}
of the quantized algebra of functions $A_q(\mathfrak{g})$ \cite{D86, RTF} 
for $\mathfrak{g}$ of type $C$.
It was done by succeeding the approach in \cite{KV}, where 
the solution of the tetrahedron equation \cite{BS} 
was obtained (up to misprint)
as the intertwiner for $A_q(\mathfrak{g})$-modules with 
$\mathfrak{g}$ of type $A$.

In this paper we extend such results to $\mathfrak{g}$ of type $B$.
The algebra $A_q(B_n)$ is formulated following \cite{RTF}, and 
its fundamental representations $\pi_i (i=1,2,3)$ 
are presented for $n=3$.
According to the general theory \cite{So1},
one has the equivalence of the tensor products
$\pi_{121} \simeq \pi_{212}$ and 
$\pi_{2323} \simeq \pi_{3232}$ 
$(\pi_{i_1,  \ldots,  i_r } := \pi_{i_1} \otimes \cdots \otimes \pi_{i_r})$
according to the Coxeter relations 
$s_1s_2s_1 = s_2s_1s_2$ and $s_2s_3s_2s_3 = s_3s_2s_3s_2$
among the simple reflections in the Weyl group $W(B_3)$.
We determine their intertwiners $\mathscr{S}$ and $\mathscr{J}$ explicitly  
and show that they yield a new solution to the 
3D reflection equation, confirming a conjecture in \cite{KO2}.
In fact, our strategy is to attribute the analysis to 
the $A_q(C_3)$ case \cite{KO2}. 
The key to this is the embedding $A_q(B_2) \hookrightarrow A_q(C_2)$
in Theorem \ref{pr:embed}.
It enable us to bypass 
the complicated intertwining relation (\ref{psi1})
for type $B$ and to utilize the result on the type $C$ case.
The embedding originates in a formulation of 
$A_q(\mathfrak{g})$ by the $RTT$ relations \cite{RTF}
and the fact that the vector representation of $C_2$ is the spinor one for $B_2$.

The paper is organized as follows.
In Section \ref{sec:rep} we recall the Hopf algebra 
$A_q(B_n)$ following \cite{RTF}.
The embedding $A_q(B_2) \hookrightarrow A_q(C_2)$ in Theorem 
\ref{pr:embed} will be the key to the proof of our main Theorem \ref{th:main}.
In Section \ref{sec:fr} we present the fundamental representations 
$\pi(i=1,2,3)$ of $A_q(B_3)$.
The $\pi_2$ and $\pi_3$ 
corresponding to the subalgebra $A_q(B_2)$ actually come from 
the fundamental representations of 
$A_q(C_2)$ through the embedding (Remark \ref{re:sub}).  
In Section \ref{sec:F} we formulate the intertwiners along the scheme 
parallel to $A_q(C_3)$ \cite{KO2} and determine them.
In Section \ref{sec:3d} 
we explain that these intertwiners yield a new solution 
to the 3D reflection equation, which is distinct from the 
type $C$ case in \cite{KO2}. 

\section{Algebra $A_q(B_n)$}\label{sec:rep}

Let $N=2n+1$ with $n \in \Z_{\ge 2}$.
We define $A_q(B_n)$ following \cite{RTF}, where it was denoted by 
$\mathrm{Fun(SO}_q(N))$.
$A_q(B_n)$ is a Hopf algebra generated by 
$t_{ij} \,(1\le i,j\le N)$ 
with the relations 
\begin{align}
&R_{ij,mp}t_{mk}t_{pl} 
= t_{jp}t_{im}R_{mp,kl},
\;\;\,
C_{jk}C_{lm}t_{ij}t_{lk} 
= C_{ij}C_{kl}t_{kj}t_{lm} = \delta_{im},
\label{re1}
\end{align}
where the repeated indices are summed over
$\{1,2,\ldots, N\}$.
The structure constants are specified by 
$C_{ij}= \delta_{i,N+1-j}q^{\varrho_j}$ with 
$(\varrho_1,\ldots, \varrho_N)
= (2n-1,,\ldots, 3, 1,0,-1,-3,\ldots,-2n+1)$,
and 
$\sum_{i,j,m,l}R_{ij,ml}E_{im}\otimes E_{jl}
= q^2\lim_{x\rightarrow \infty}x^{-2}R(x)|_{k=q^{-2}}$,
where $R(x)$ is the quantum $R$ matrix \cite{Baz,J2} for 
the vector representation of $U_q(B^{(1)}_n)$ given in \cite[eq.(3.6)]{J2}.
The former one in (\ref{re1}) is the so called $RTT$ relation on the 
generators $T=(t_{ij})$.
The coproduct is given by
$\Delta(t_{ij}) = \sum_k t_{ik} \otimes t_{kj}$.
We omit the antipode and counit as they are not necessary in this paper.

Reflecting the equivalence $U_q(B_2) \simeq U_q(C_2)$,
the simplest case $A_q(B_2)$ is related to $A_q(C_2)$.
The latter was introduced in \cite{RTF}  and detailed in 
\cite[sec. 3.1]{KO2}.
where it was denoted by $A_q(\mathrm{Sp}_4)$.
The $A_q(C_2)$ is the Hopf algebra defined similarly to $A_q(B_2)$ 
by replacing the structure constants $R_{ij,mp}$ and $C_{jk}$
by those associated with the quantum $R$ matrix of the vector representation of
$U_q(C^{(1)}_2)$.
The generators $t_{ij} \,(1\le i,j \le 4)$ of $A_q(C_2)$ (\cite[sec. 3.1]{KO2}) 
will be denoted by 
$s_{ij}$ here for distinction. 
The vector representation of $U_q(C_2)$ is the spinor 
one in terms of $U_q(B_2)$.
This fact is reflected in the following.

\begin{theorem}\label{pr:embed}
There is an embedding of the Hopf algebra
$\iota: A_q(B_2) \hookrightarrow A_q(C_2)$ which 
maps the generators $\{t_{ij} \mid 1 \le i,j \le 5\}$ as follows:
\begin{align*}
t_{ij} &\mapsto (-1)^{\delta_{i 5}+\delta_{j 5}}(\sqrt{r})^{\delta_{i3}}
(s_{\xi_i \xi_j}s_{\eta_i \eta_j}-q s_{\xi_i \eta_j}s_{\eta_i \xi_j}) \quad (j\neq 3),\\
t_{i 3}&\mapsto (-1)^{\delta_{i5}}\sqrt{r}\, 
(-q^{-1}s_{\xi_i 1}s_{\eta_i 4} + q s_{\xi_i 4}s_{\eta_i 1})\quad (i \neq 3),\\
t_{33} &\mapsto s_{22}s_{33}-q s_{21}s_{34} + q s_{24}s_{31} - q^2 s_{23}s_{32},
\end{align*}
where $r= 1+q^2$ and $\xi_i, \eta_i$ are given by
$\left({\xi_1 \xi_2 \xi_3 \xi_4 \xi_5 \atop 
\eta_1 \eta_2 \eta_3 \eta_4 \eta_5}\right)=
\left({1 1 2 2 3\atop 2 3 3 4 4}\right)$.
\end{theorem}

We have checked that under $\iota$, 
all the relations on $\{t_{ij}\}$ in $A_q(B_2)$ are 
guaranteed by those on $\{s_{ij}\}$ in $A_q(C_2)$ by tedious but direct 
calculations.
We note that in the different formulation of $A_q(\mathfrak{g})$ of \cite{Ka}, 
one has the isomorphism $A_q(B_2)\simeq A_q(C_2)$ by definition 
rather than the embedding.

\section{Fundamental representations of $A_q(B_3)$}\label{sec:fr}
Let $\mathrm{Osc}_q = 
\langle {{\bf 1}, \rm{\bf a}}^+, {\rm{\bf a}}^-, {\rm{\bf k}} \rangle$ 
be the $q$-oscillator algebra, i.e. 
an associative algebra with the center ${\bf 1}$ and the relations
\begin{equation*}
\ok \,\ap = q\,\ap \,\ok,\;\;
\ok\,\am = q^{-1}\am\,\ok,\;\;
\am \ap = {\bf 1}\!-\!q^2 \,\ok^2,\;\;
\ap \am = {\bf 1}\!-\!{\rm{\bf k}}^2.
\end{equation*}
It has a representation on 
the Fock space $\mathcal{F}_q = \oplus_{m\ge 0}{\mathbb C}(q)|m\rangle$:
\begin{align*}
{\bf 1}|m\rangle = |m\rangle,\;
\ok |m\rangle = q^m |m\rangle,\;
\ap |m\rangle = |m+1\rangle,\;
\am |m\rangle = (1-q^{2m})|m-1\rangle.
\end{align*}
We shall also use $\mathrm{Osc}_q$ with $q$ replaced by $q^2$.
For distinction it is denoted by 
$\mathrm{Osc}_{q^2} 
= \langle {{\bf 1}, \Ap, \Am, \OK} \rangle$, 
which acts on $\mathcal{F}_{q^2}$.

{}From now on we consider $A_q(B_3)$ which includes 
$A_q(B_2)$ as a subalgebra.
Set $(q_1,q_2,q_3)=(q^2,q^2,q)$, which reflects the 
squared root length of the simple roots of $B_3$.
Consider the maps 
$\pi_i \,(i\!=\!1,2,3):\, A_q(B_3) \rightarrow \mathrm{Osc}_{q_i}$
that send the generators $(t_{ij})_{1 \le i,j \le 7}$ 
to the following:
\begin{align}
&\pi_1:
\begin{pmatrix}
\mu_1\Am & \alpha_1\OK  &  & & & &\\
\beta_1 \OK & \nu_1\Ap &  & & & & \\
 & & \kappa_1{\bf 1} & & & &\\
 & & & \sigma_1{\bf 1}&  & & \\
 & & & &  \kappa^{-1}_1{\bf 1}& &\\
 &&&&& \nu_1^{-1}\Am & q^2\beta^{-1}_1\OK\\
&&&& & q^2\alpha^{-1}_1\OK & \mu^{-1}_1\Ap
 \end{pmatrix},\label{pi1}\\
&\pi_2:
\begin{pmatrix}
\kappa_2{\bf 1}  & & & & & &\\
&\mu_2\Am & \alpha_2\OK    & & & &\\
&\beta_2 \OK & \nu_2\Ap   & & & & \\
 & & & \sigma_2{\bf 1}&  & & \\
 &&&& \nu_2^{-1}\Am & q^2\beta^{-1}_2\OK&\\
&&& & q^2\alpha^{-1}_2\OK & \mu^{-1}_2\Ap&\\
 & & & &  & & \kappa^{-1}_2{\bf 1}
 \end{pmatrix}, \label{pi2}\\
&\pi_3:
\begin{pmatrix}
\kappa_{31}{\bf 1}  & & & & & &\\
& \kappa_{32}{\bf 1}  & & & & &\\
& & \mu_3 (\am)^2 & \alpha_3\ok \,\am & -(r\mu_3)^{-1}\alpha_3^2 \ok^2& &\\
& & -r\alpha^{-1}_3\mu_3 \am \, \ok  & \am\ap-\ok^2& 
-(q\mu_3)^{-1}\alpha_3 \ok\, \ap & &\\
& & -r(q\alpha^{-1}_3)^2\mu_3 \ok^2 & r\alpha^{-1}_3\ok \, \ap 
& \mu_3^{-1}(\ap)^2& &\\
& & & &  & \kappa^{-1}_{32}{\bf 1} & \\
& & & &  & & \kappa^{-1}_{31}{\bf 1}
 \end{pmatrix}, \label{pi3}
\end{align}
where $r$ is defined in Theorem \ref{pr:embed} and blanks mean $0$.
The symbols $\alpha_i, \beta_i, \mu_i, \nu_i, \sigma_i, \kappa_i, 
\kappa_{31}, \kappa_{32}$ are parameters.
They obey the constraint 
\begin{equation}\label{paracon}
\alpha_i \beta_i = -q^2 \mu_i \nu_i,\quad
\sigma_i = \pm 1\quad (i=1,2).
\end{equation}

\begin{proposition}\label{pr:piA}
The maps $\pi_i\,(i\!=\!1,2,3)$ are naturally extended to the 
algebra homomorphisms. 
The resulting representations of $A_q(B_3)$ 
on $\mathcal{F}_{q_i}$ are irreducible.
\end{proposition}

We call $\pi_i\,(i=1,2,3)$ the {\em fundamental representations}
of $A_q(B_3)$.
It is easy to infer the fundamental representation $\pi_i$ of
$A_q(B_n)$ for general $n$.
\begin{remark}\label{re:sub}
The symbols $(t_{ij})_{2 \le i,j \le 6}$ generate $A_q(B_2)$.
The $5 \times 5$ submatrices of (\ref{pi2}) and (\ref{pi3}) 
without the first and the last rows and columns
provide the two fundamental representations of the $A_q(B_2)$.
Denote them by $\pi^{B_2}_1$ and $\pi^{B_2}_2$ respectively.
Similarly one can extract the fundamental representations of 
$A_q(C_2)$ from the $4\times 4$ submatrices of 
\cite[eq.(3.4), eq(3.5)]{KO2}.
Denote them by $\pi^{C_2}_1$ and $\pi^{C_2}_2$, respectively.
Then $\pi^{B_2}_i$ coincides with the representation induced from 
$\pi^{C_2}_{3-i}$ via Theorem \ref{pr:embed}
with a suitable adjustment of the parameters, i.e.
$\pi^{B_2}_i(f) = \pi^{C_2}_{3-i}(\iota(f))$ for any 
$f \in A_q(B_2)$.
\end{remark}

\section{Intertwiners}\label{sec:F}

The Weyl group $W(B_3) = \langle s_1, s_2, s_3\rangle$ is 
generated by the simple reflections with the relations
$
s_1^2=s_2^2=s_3^2=1,\, 
s_1s_3 = s_3s_1,\,s_1s_2s_1 = s_2s_1s_2$ and  
$s_2s_3s_2s_3 = s_3s_2s_3s_2$.
According to \cite{So1} one should have the equivalence
\begin{equation}\label{equi}
\pi_{13} \simeq \pi_{31},\quad \pi_{121} \simeq \pi_{212},\quad
\pi_{2323} \simeq \pi_{3232}
\end{equation}
under an appropriate tuning of the parameters.
Here and in what follows we use the 
shorthand $\pi_{i_1,  \ldots,  i_r } = \pi_{i_1} \otimes \cdots \otimes \pi_{i_r}$.
It is easy to see that $\pi_{13} \simeq \pi_{31}$ holds if and only if
\begin{equation}\label{par1}
\kappa_1 = \sigma_1,\quad \kappa_{31} = \kappa_{32}
\end{equation}
and the intertwiner is just the transposition
$P(x\otimes y) = y \otimes x$.
The equivalence $\pi_{121} \simeq \pi_{212}$ holds if and only if 
\begin{equation}\label{par2}
\kappa_1 = \kappa_2 = \sigma_2,\quad \alpha_1\beta_1 = \alpha_2\beta_2
\end{equation}
are further satisfied.
Assuming them we introduce the intertwiner
$\Phi  \in \mathrm{End}(\mathcal{F}_{q^2} 
\otimes \mathcal{F}_{q^2}\otimes \mathcal{F}_{q^2})$
characterized by
\begin{align*}
\pi_{212}(\Delta(f))\circ \Phi = \Phi \circ \pi_{121}(\Delta(f))
\quad (\forall f \in A_q(B_3))
\end{align*}
and the normalization
$\Phi (|0\rangle \otimes|0\rangle \otimes|0\rangle) 
=|0\rangle \otimes|0\rangle \otimes|0\rangle$.
Set $S=\Phi P_{13}$, where
$P_{13}(x\otimes y \otimes z) =z \otimes y \otimes x$.
The $S$ is regarded as a matrix 
$S = (S^{abc}_{ijk})$ whose elements are specified by
$
S(|i\rangle \otimes |j\rangle \otimes |k\rangle) = 
\sum_{a,b,c \in {\mathbb Z}_{\ge 0}} S^{a b c}_{i j k}
|a\rangle \otimes |b\rangle \otimes |c\rangle$.
We use the notation $\delta^a_i=1 (a=i), \;\delta^a_i=0 (a\neq i)$ and 
\begin{align*}
(q)_i = \prod_{j=1}^i(1-q^j),\quad
\left[i_1,\ldots, i_r \atop j_1, \ldots, j_s\right]_{\!q} =
\begin{cases} 
\frac{\prod_{k=1}^r(q)_{i_k}}{\prod_{k=1}^s(q)_{j_k}} & 
\forall i_k, j_k \in \Z_{\ge 0},\\
0 & \text{otherwise}.
\end{cases}
\end{align*}

The following result confirms the conjecture stated in 
\cite[eq.(4.2)]{KO2}.
\begin{theorem}\label{th:121}
Under (\ref{paracon}), (\ref{par1}) and (\ref{par2}), 
the following formulas are valid:
\begin{align*}
S^{abc}_{ijk} &= (-\alpha_1\beta_1q^{-2})^j 
\mu_1^{a-j+k}\mu_2^{b-a-k}\sigma_1^{b+j}
\mathscr{S}^{abc}_{ijk},\\
\mathscr{S}^{abc}_{ijk} &=\delta_{i+j}^{a+b}\delta_{j+k}^{b+c}
\!\!\sum_{\lambda+\mu=b}\!\!\!(-1)^\lambda
q^{2i(c-j)+2(k+1)\lambda+2\mu(\mu-k)}
\!\left[{i,j,c+\mu \atop \mu,\lambda,i-\mu,j-\lambda,c}\right]_{q^4}.
\end{align*}
\end{theorem}
This $\mathscr{S}^{abc}_{ijk}$ is equal to 
 $\mathscr{R}^{abc}_{ijk}$ in \cite[eq.(2.20)]{KO2} with $q$ replaced by $q^2$.
Our notation here is taken consistently with \cite[eq.(4.2)]{KO2} .
This $\mathscr{S}$ coincides with the 3D $R$ matrix \cite{BS}
which follows as the intertwiner of $A_{q^2}(\mathrm{SL}_3)$-modules \cite{KV}.
See \cite[sec.2]{KO2} for an exposition.
It satisfies the tetrahedron equation 
\begin{equation}\label{tehat}
\mathscr{S}_{356}\mathscr{S}_{246}\mathscr{S}_{145}\mathscr{S}_{123}
=\mathscr{S}_{123}\mathscr{S}_{145}\mathscr{S}_{246}\mathscr{S}_{356}
\end{equation}
and $\mathscr{S} = \mathscr{S}^{-1}$, 
$\mathscr{S}_{ijk}=\mathscr{S}_{kji}$.

Let us turn to   
$\pi_{2323}\simeq \pi_{3232}$ in (\ref{equi}).
It holds if and only if 
\begin{align}\label{par3}
\kappa_{31} = \kappa_{32} = \pm 1,\quad 
\alpha_2\beta_2 = \pm q^2.
\end{align}
The constraints 
(\ref{paracon}), (\ref{par1}),  (\ref{par2}) and (\ref{par3}) 
are summarized as 
\begin{align*}
\sigma = \kappa_i=\sigma_i,\;\;
\rho = \kappa_{3i} ,\;\;
\alpha_i\beta_i = -\varepsilon q^2,\;\;
\mu_i\nu_i = \varepsilon\;\;\;(i=1,2)
\end{align*}
in terms of the three independent sign factors 
$\sigma, \rho,\varepsilon \in \{1,-1\}$.

Let 
$\Xi: \mathcal{F}_{q^2} \otimes \mathcal{F}_q\otimes 
\otimes \mathcal{F}_{q^2}\otimes \mathcal{F}_q
\rightarrow 
\mathcal{F}_q\otimes 
\otimes \mathcal{F}_{q^2}\otimes \mathcal{F}_q
\otimes \mathcal{F}_{q^2}$
be the intertwiner characterized by
\begin{align}
\pi_{3232}(\Delta(f))\circ \Xi = \Xi \circ \pi_{2323}(\Delta(f))
\quad (\forall f \in A_q(B_3))\label{psi1}
\end{align}
and the normalization
$\Xi (|0\rangle \otimes|0\rangle \otimes|0\rangle\otimes|0\rangle) 
=|0\rangle \otimes|0\rangle \otimes|0\rangle\otimes|0\rangle$.
Introduce further $J=\Xi \,P_{14}P_{23}
\in \mathrm{End}(\mathcal{F}_q\otimes 
\otimes \mathcal{F}_{q^2}\otimes \mathcal{F}_q
\otimes \mathcal{F}_{q^2})$, where
$P_{14}P_{23}(x\otimes y \otimes z \otimes w) 
=w \otimes z \otimes y \otimes x$.
The $J$ is regarded as a matrix 
$J = (J^{abcd}_{ijkl})$ whose elements are specified by
\begin{equation*}
J(|i\rangle \otimes |j\rangle \otimes |k\rangle\otimes |l\rangle) = 
\sum_{a,b,c,d \in {\mathbb Z}_{\ge 0}} J^{a b c d}_{i j k l}
|a\rangle \otimes |b\rangle \otimes |c\rangle\otimes |d\rangle.
\end{equation*}
Setting 
$
J^{a b c d}_{i j k l} = \varepsilon^{b+i+l}\rho^{b+c}
(\sigma \mu_2)^{c-k}\mu_3^{b-j}\mathscr{J}^{a b c d}_{i j k l}$,
it can easily be checked that 
$\mathscr{J}^{a b c d}_{i j k l}$ depends only on $q$.
In this sense we call 
$\mathscr{J} = (\mathscr{J}^{a b c d}_{i j k l})$ 
the parameter-free part of the intertwiner for 
$A_q(B_3)$.

Let $\mathscr{K} = (\mathscr{K}^{abcd}_{ijkl})$ be
the parameter-free part in the similar sense of the intertwiner for the
quantized function algebra $A_q(C_3)$.
The elements $\mathscr{K}^{abcd}_{ijkl}$ are polynomials in $q$ and given 
explicitly in \cite[Th.3.4]{KO2}.
The $\mathscr{J}$ and $\mathscr{K}$ will simply be referred to as 
type $B$ and type $C$, respectively.
Now we present the main result of the paper,
which confirms the conjecture in \cite[eq.(4.1)]{KO2}.

\begin{theorem}\label{th:main}
The type $B$ and $C$ intertwiners are related by the transposition of the components as 
$\mathscr{J} = P_{14}P_{23}\mathscr{K}P_{14}P_{23}$, i.e.
$
\mathscr{J}^{a b c d}_{i j k l} = \mathscr{K}^{d c b a}_{l k j i}
$
holds.
\end{theorem}

\begin{proof}
The $\mathscr{J}$ and $\mathscr{K}$ are characterized as the
intertwiners of the representations of the subalgebras $A_q(B_2)$ and $A_q(C_2)$,
respectively.  
Thus the assertion follows from Theorem \ref{pr:embed} and 
Remark \ref{re:sub}.
\end{proof}

As a corollary of \cite[Th.3.5]{KO2}, 
$\mathscr{J}^{a b c d}_{i j k l}$ is a polynomial in $q^2$ 
vanishing unless $(a+2b+c, b+c+d)=(i+2j+k, j+k+l)$.
Moreover $\mathscr{J}$ has the birational and combinatorial counterparts
as shown in \cite[Table 1]{KO2}.

\begin{example}
The following is the list of all the nonzero $\mathscr{J}^{1102}_{ijkl}$.
\begin{alignat*}{2}
\mathscr{J}^{1102}_{3003} &= q^8 (1 - q^6)(1 - q^{12}), \qquad &
\mathscr{J}^{1102}_{2012} &=q^4(1 - q^4) (1 - q^2 + q^4 - q^6 - q^{10}),\\
\mathscr{J}^{1102}_{1021} &=-q^6(1-q^6), &
\mathscr{J}^{1102}_{0030} &=-q^2(1-q^6),\\
\mathscr{J}^{1102}_{1102} &=q^2 (1 - q^8 + q^{14}), &
\mathscr{J}^{1102}_{0111} &=1 - q^4 + q^{10}.
\end{alignat*}
\end{example}

\section{3D reflection equation}\label{sec:3d}
The intertwiners $\mathscr{S}$ and $\mathscr{J}$ 
yield a solution of the 3D reflection equation proposed by 
Isaev and Kulish \cite{IK}.
Since its derivation is the same as the type $C$ case in 
\cite[sec.3.6]{KO2}, we shall only describe the outline. 

Let $w_0$ be the longest element of the Weyl group $W(B_3)$ 
and consider the two reduced expressions
$w_0 = s_1s_2s_3s_2s_1s_2s_3s_2s_3 
= s_3s_2s_3s_2s_1s_2s_3s_2s_1$.
According to \cite[Th. 5.7, Cor. 1]{So1}, we have the equivalence 
$\pi_{123212323} \simeq \pi_{323212321}$ of the 
representations of $A_q(B_3)$.
The intertwiner for this can be constructed by 
composing the `basic ones' $P, \Phi$ and $\Xi$ for (\ref{equi}).
Moreover there are two ways to do so, the results of which must coincide
since $\pi_{123212323} \simeq \pi_{323212321}$ is irreducible.
This postulate leads to the equality
\begin{align}\label{3dref}
\mathscr{S}_{456}\mathscr{S}_{489}
\mathscr{J}_{3579}\mathscr{S}_{269}\mathscr{S}_{258}
\mathscr{J}_{1678}\mathscr{J}_{1234}
=\mathscr{J}_{1234}
\mathscr{J}_{1678}
\mathscr{S}_{258}\mathscr{S}_{269}
\mathscr{J}_{3579}\mathscr{S}_{489}\mathscr{S}_{456}
\end{align}
in $\mathrm{End}(\mathcal{F}_{q}\otimes \mathcal{F}_{q^2}
\otimes \mathcal{F}_{q}\otimes \mathcal{F}_{q^2}
\otimes \mathcal{F}_{q^2}\otimes \mathcal{F}_{q^2}
\otimes \mathcal{F}_{q}\otimes \mathcal{F}_{q^2}
\otimes \mathcal{F}_{q^2})$,
where the indices 
signify the positions of the tensor components on which the 
intertwiners act non trivially.
The equation (\ref{3dref}) is the 3D reflection (or 
``tetrahedron reflection") equation introduced in \cite{IK}
without a spectral parameter.
In view of Theorem \ref{th:main},
one may substitute $\mathscr{J}_{ijkl} = \mathscr{K}_{lkji}$.
The resulting relation confirms the conjecture stated in \cite[eq.(4.3)]{KO2}.

A pair $(\mathscr{S}, \mathscr{J})$ is called a solution to the 3D reflection 
equation if it satisfies (\ref{3dref}) and $\mathscr{S}$ satisfies 
the tetrahedron equation (\ref{tehat}) by itself.
Thus we have established a new solution of the 3D reflection equation
distinct from the 
previous one $(\mathscr{R}, \mathscr{K})$ in \cite{KO2}.
The two solutions $(\mathscr{S}, \mathscr{J})$
and $(\mathscr{R}, \mathscr{K})$ are associated with 
$A_q(B_3)$ and $A_q(C_3)$, respectively.
Although they are simply related as
$(\mathscr{S}, \mathscr{J}) = 
(\mathscr{R}|_{q\rightarrow q^2}, P_{14}P_{23}\mathscr{K}P_{14}P_{23})$,
the fact that these replacements do keep 
(\ref{3dref}) is highly nontrivial.
The main achievement of this paper is to have established it 
by exploiting the embedding in Theorem \ref{pr:embed}. 

Let us finish by giving a diagram for the 3D reflection equation
\begin{align}\label{zueq}
\mathscr{S}_{489}
\mathscr{J}_{3579}\mathscr{S}_{269}\mathscr{S}_{258}
\mathscr{J}_{1678}\mathscr{J}_{1234}\mathscr{S}_{654}
=\mathscr{S}_{654}\mathscr{J}_{1234}
\mathscr{J}_{1678}
\mathscr{S}_{258}\mathscr{S}_{269}
\mathscr{J}_{3579}\mathscr{S}_{489}.
\end{align}
This is $\mathscr{S}_{456}^{-1}\times {\rm eq. }(\ref{3dref})
\times \mathscr{S}_{456}$
with the properties mentioned after (\ref{tehat}) applied.
It is depicted in Fig.1.

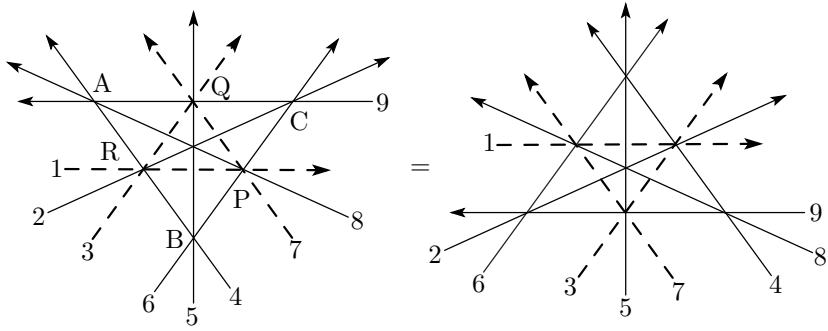
\begin{figure}[h]
\unitlength 0.1in
\begin{picture}(9.6000, 17)(24,-23.0)
%
\special{pn 8}%
\special{pa 3246 1974}%
\special{pa 4206 654}%
\special{fp}%
\special{sh 1}%
\special{pa 4206 654}%
\special{pa 4150 696}%
\special{pa 4174 698}%
\special{pa 4182 720}%
\special{pa 4206 654}%
\special{fp}%
%
\special{pn 13}%
\special{pa 4236 2004}%
\special{pa 3468 948}%
\special{da 0.070}%
\special{sh 1}%
\special{pa 3468 948}%
\special{pa 3490 1014}%
\special{pa 3498 992}%
\special{pa 3522 990}%
\special{pa 3468 948}%
\special{fp}%
%
\special{pn 13}%
\special{pa 3744 1998}%
\special{pa 4512 942}%
\special{da 0.070}%
\special{sh 1}%
\special{pa 4512 942}%
\special{pa 4456 984}%
\special{pa 4480 986}%
\special{pa 4488 1008}%
\special{pa 4512 942}%
\special{fp}%
%
\special{pn 13}%
\special{pa 3318 1308}%
\special{pa 4690 1302}%
\special{da 0.070}%
\special{sh 1}%
\special{pa 4690 1302}%
\special{pa 4624 1282}%
\special{pa 4638 1302}%
\special{pa 4624 1322}%
\special{pa 4690 1302}%
\special{fp}%
%
\special{pn 8}%
\special{pa 4926 1662}%
\special{pa 3078 1662}%
\special{fp}%
\special{sh 1}%
\special{pa 3078 1662}%
\special{pa 3144 1682}%
\special{pa 3130 1662}%
\special{pa 3144 1642}%
\special{pa 3078 1662}%
\special{fp}%
%
\special{pn 8}%
\special{pa 4752 1992}%
\special{pa 3792 672}%
\special{fp}%
\special{sh 1}%
\special{pa 3792 672}%
\special{pa 3814 738}%
\special{pa 3822 716}%
\special{pa 3846 714}%
\special{pa 3792 672}%
\special{fp}%
%
\special{pn 8}%
\special{pa 3990 2094}%
\special{pa 3990 570}%
\special{fp}%
\special{sh 1}%
\special{pa 3990 570}%
\special{pa 3970 638}%
\special{pa 3990 624}%
\special{pa 4010 638}%
\special{pa 3990 570}%
\special{fp}%
%
\special{pn 8}%
\special{pa 4968 1872}%
\special{pa 3186 1062}%
\special{fp}%
\special{sh 1}%
\special{pa 3186 1062}%
\special{pa 3238 1108}%
\special{pa 3234 1084}%
\special{pa 3254 1072}%
\special{pa 3186 1062}%
\special{fp}%
%
\special{pn 8}%
\special{pa 3042 1860}%
\special{pa 4824 1050}%
\special{fp}%
\special{sh 1}%
\special{pa 4824 1050}%
\special{pa 4754 1060}%
\special{pa 4774 1072}%
\special{pa 4772 1096}%
\special{pa 4824 1050}%
\special{fp}%
\put(32.7500,-12.9600){\makebox(0,0){1}}%
\put(29.9300,-18.8400){\makebox(0,0){2}}%
\put(37.0100,-20.4600){\makebox(0,0){3}}%
\put(47.7500,-20.4000){\makebox(0,0){4}}%
\put(39.8900,-21.6000){\makebox(0,0){5}}%
\put(32.2100,-20.3400){\makebox(0,0){6}}%
\put(42.6500,-20.7600){\makebox(0,0){7}}%
\put(50.0900,-18.9000){\makebox(0,0){8}}%
\put(49.8500,-16.6200){\makebox(0,0){9}}%
%
\special{pn 8}%
\special{pa 1524 2076}%
\special{pa 2484 756}%
\special{fp}%
\special{sh 1}%
\special{pa 2484 756}%
\special{pa 2428 798}%
\special{pa 2452 800}%
\special{pa 2460 822}%
\special{pa 2484 756}%
\special{fp}%
%
\special{pn 8}%
\special{pa 966 1662}%
\special{pa 2748 852}%
\special{fp}%
\special{sh 1}%
\special{pa 2748 852}%
\special{pa 2678 862}%
\special{pa 2698 874}%
\special{pa 2696 898}%
\special{pa 2748 852}%
\special{fp}%
%
\special{pn 8}%
\special{pa 2538 1686}%
\special{pa 756 876}%
\special{fp}%
\special{sh 1}%
\special{pa 756 876}%
\special{pa 808 922}%
\special{pa 804 898}%
\special{pa 824 886}%
\special{pa 756 876}%
\special{fp}%
%
\special{pn 8}%
\special{pa 1728 2136}%
\special{pa 1728 612}%
\special{fp}%
\special{sh 1}%
\special{pa 1728 612}%
\special{pa 1708 680}%
\special{pa 1728 666}%
\special{pa 1748 680}%
\special{pa 1728 612}%
\special{fp}%
%
\special{pn 8}%
\special{pa 1920 2058}%
\special{pa 960 738}%
\special{fp}%
\special{sh 1}%
\special{pa 960 738}%
\special{pa 982 804}%
\special{pa 990 782}%
\special{pa 1014 780}%
\special{pa 960 738}%
\special{fp}%
%
\special{pn 13}%
\special{pa 1200 1800}%
\special{pa 1968 744}%
\special{da 0.070}%
\special{sh 1}%
\special{pa 1968 744}%
\special{pa 1912 786}%
\special{pa 1936 788}%
\special{pa 1944 810}%
\special{pa 1968 744}%
\special{fp}%
%
\special{pn 13}%
\special{pa 2244 1794}%
\special{pa 1476 738}%
\special{da 0.070}%
\special{sh 1}%
\special{pa 1476 738}%
\special{pa 1498 804}%
\special{pa 1506 782}%
\special{pa 1530 780}%
\special{pa 1476 738}%
\special{fp}%
%
\special{pn 8}%
\special{pa 2664 1080}%
\special{pa 816 1080}%
\special{fp}%
\special{sh 1}%
\special{pa 816 1080}%
\special{pa 882 1100}%
\special{pa 868 1080}%
\special{pa 882 1060}%
\special{pa 816 1080}%
\special{fp}%
%
\special{pn 13}%
\special{pa 1056 1434}%
\special{pa 2428 1440}%
\special{da 0.070}%
\special{sh 1}%
\special{pa 2428 1440}%
\special{pa 2362 1420}%
\special{pa 2376 1440}%
\special{pa 2362 1460}%
\special{pa 2428 1440}%
\special{fp}%
\put(10.1300,-14.2800){\makebox(0,0){1}}%
\put(9.1700,-16.8000){\makebox(0,0){2}}%
\put(11.7500,-18.6600){\makebox(0,0){3}}%
\put(19.4900,-21.0000){\makebox(0,0){4}}%
\put(17.2100,-22.0800){\makebox(0,0){5}}%
\put(14.9300,-21.3600){\makebox(0,0){6}}%
\put(22.6100,-18.5400){\makebox(0,0){7}}%
\put(25.8500,-17.0400){\makebox(0,0){8}}%
\put(27.1700,-10.8600){\makebox(0,0){9}}%
\put(29.1500,-14.3400){\makebox(0,0){$=$}}%
\put(12.5000,-9.8000){\makebox(0,0){A}}%
\put(16.3000,-18.0000){\makebox(0,0){B}}%
\put(22.8000,-11.8000){\makebox(0,0){C}}%
\put(19.8000,-15.9000){\makebox(0,0){P}}%
\put(18.7000,-10.1000){\makebox(0,0){Q}}%
\put(12.9000,-13.4000){\makebox(0,0){R}}%
\end{picture}%
\caption{The 2D projection diagram of the 3D reflection equation (\ref{zueq}).}
\end{figure}
\noindent
To draw the LHS, 
first put three lines 8, 5 and 2 that intersect at one point,  the center.
Take generic points A, B and C  on every second 
half-infinite line emanating from the center along 8, 5 and 2.
The points P, Q and R are the crossings of the triangle 
ABC with the lines 8, 5 and 2.
The broken lines are formed by connecting P, Q and R. 
The RHS is obtained by changing A, B and C in such a way that
the lines 6, 9  and 4 are shifted and the triangle ABC is reversed.

There is a natural geometric interpretation of Fig.1.
Let us call the plane containing all the lines 
the {\em boundary plane}.
Fig. 1 shows the projection of the three planes reflected by the boundary plane.
A physical interpretation is the world sheets of three 
straight strings in 3D exhibiting boundary reflections.

The intersections of the three and four arrows correspond to $\mathscr{S}$ and 
$\mathscr{J}$  in (\ref{zueq}), respectively.  
The indices $1,3$ and $7$ are attached only to $\mathscr{J}$,
hence represent the boundary degrees of freedom.
In Fig.1 they are depicted by the broken lines.
The other solid lines are the projection of the 
configurations (or events) in 3D explained in what follows
onto the boundary plane.

To each broken line, associate a pair of half-planes in 3D
sharing the broken line as the common boundary.
They represent the world sheet of a straight string moving in 3D
and reflected by the boundary plane exactly at the broken line.
The six in total half-planes should be located 
on the same side of the boundary plane. 
A broken line is like a spine of an open (by a generic angle) book and the 
half-planes are like the front and back of the book.
The books should be ``upright" on the boundary plane since
the incident and reflection angles coincide.
There are three such books with spines $1,3$ and $7$.
At the intersection R of the spines $1$ and $3$,
the fronts and backs of the two books generate 
the four intersecting half-lines on them.
Their projections are the solid lines $2$ and $4$.
The crossing of these four lines $1, 2, 3, 4$ yields $\mathscr{J}_{1234}$.
The other $\mathscr{J}$ can be understood similarly.
The intersections of three planes without a spine 
take place off the boundary plane and 
correspond to  three string scatterings 
encoded by $\mathscr{S}$.
By an elementary geometry one can prove that 
the 2D projection of the 3D configuration explained so far 
exactly forms the pattern depicted in Fig.1.
Finally the orientations of the lines indicated by the arrows 
signify the time ordering of the scattering and reflection events.
They specify the order of the products of various
$\mathscr{S}$ and $\mathscr{J}$ in (\ref{zueq}).
\section*{Acknowledgments}
The authors thank the organizers of 
The XXIX International Colloquium on Group-Theoretical Methods in Physics,
August 20-26, 2012 at Chern Institute of Mathematics, Tianjin, China
for invitation and warm hospitality.
This work is supported by Grants-in-Aid for
Scientific Research No.~23340007, No.~24540203 and 
No.~23654007 from JSPS.

\bibliographystyle{ws-procs9x6}

\begin{thebibliography}{99}

\bibitem{Ch84}
I.~V.~Cherednik,
{\it Theor. Math. Phys.} {\bf 61}, 977 (1984).

\bibitem{Sk88}
E.~K.~Sklyanin,
{\it J. Phys. A: Math. Gen.} {\bf 21}, 2375 (1988).

\bibitem{Bax}
R.~ J.~ Baxter,
\textit{Exactly solved models in statistical mechanics} (Dover, 2007).

\bibitem{IK}
A.~P.~Isaev and P.~P.~Kulish,
{\it Mod. Phys. Lett. A}~{\bf 12}, 427 (1997).

\bibitem{Zam80}
A.~B.~Zamolodchikov,
{\it Soviet Phys. JETP} {\bf 52}, 325 (1980).

\bibitem{KO2}
A.~Kuniba and M.~Okado,
{\it J. Phys. A: Math. Theor.}  {\bf 45} 465206, 27pp (2012).

\bibitem{So1}Y.~ S.~ Soibelman,
{\it Leningrad Math. J.} {\bf 2}, 161 (1991).

\bibitem{D86}
V.~G.~Drinfeld,
Quantum groups,
in {\it Proceedings of the International Congress of Mathematicians}, pp798--820
(Berkeley, 1986).

\bibitem{RTF}N.~Yu.~Reshetikhin, L.~A.~Takhtadzhyan and L.~D.~Faddeev, 
{\it Leningrad Math. J.} {\bf 1}, 193 (1990).

\bibitem{KV}
M.~M.~Kapranov and V.~A.~Voevodsky,
{\it Proc. Symposia in Pure Math.} {\bf 56}, 177 (1994).

\bibitem{BS}
V.~V.~Bazhanov and  S.~M.~Sergeev,
{\it J. Phys. A: Math. Theor.} {\bf 39}, 3295 (2006).

\bibitem{Baz}
V. V. Bazhanov,
{\it Phys. Lett. B}~{\bf 159}, 321 (1985).

\bibitem{J2}
M.~Jimbo,
{\it Commun. Math. Phys.} {\bf 102}, 537 (1986).

\bibitem{Ka}
M.~Kashiwara,
{\it Duke Math. J.} {\bf 69}, 455 (1993).

\end{thebibliography}

\end{document}